\pdfoutput=1
\def\titlerunning{Gottesman Types for Quantum Programs}
\def\authorrunning{Rand, Sundaram, Singhal \& Lackey}

\documentclass[copyright,creativecommons]{eptcs}
\providecommand{\event}{QPL 2020} %
\usepackage[a-1b]{pdfx}
\usepackage{underscore}           %
\usepackage{amsthm}
\newtheorem{proposition}[]{Proposition}

\newtheorem{corollary}{Corollary}

\usepackage{amsmath,amsfonts}
\usepackage{tikz}
\usetikzlibrary{shapes}
\usetikzlibrary{decorations.pathreplacing}
\usepackage{changepage}
\usepackage{braket}
\usepackage{physics}
\usepackage{mathpartir}
\usepackage{xcolor}
\usepackage{multirow}
\usepackage{tabu}
\usepackage{xspace}
\usepackage{hyperref}
\usepackage{cleveref}
\usepackage{paralist}
\usepackage[normalem]{ulem}

\usepackage{bussproofs}

\crefformat{section}{\S#2#1#3}
\crefformat{subsection}{\S#2#1#3}
\crefformat{subsubsection}{\S#2#1#3}
\crefrangeformat{section}{\S\S#3#1#4 to~#5#2#6}
\crefmultiformat{section}{\S\S#2#1#3}{ and~#2#1#3}{, #2#1#3}{ and~#2#1#3}

\usepackage{xfrac}
\usepackage{stmaryrd,bbold}

\usepackage{alltt}
\usepackage{listings,lstcoq}
\definecolor{ltblue}{rgb}{0,0.4,0.4}
\definecolor{dkblue}{rgb}{0,0.1,0.6}
\definecolor{dkgreen}{rgb}{0,0.35,0}
\definecolor{dkviolet}{rgb}{0.3,0,0.5}
\definecolor{dkred}{rgb}{0.5,0,0}
\newcommand{\comp}{\mathrel{\$}}

\newcommand{\ctrl}[2]{\draw[fill=black] (#1,#2) circle [radius=0.12];}

\newcommand{\unitary}[3]{\draw[fill=white] (#2-0.4,#3-0.4) rectangle node {\texttt{#1}} (#2+0.4,#3+0.4);}

\newcommand{\qwire}{\ensuremath{\mathcal{Q}\textsc{wire}}\xspace}

\newcommand{\type}[1]{\ensuremath{\mathbf{#1}}\xspace}

\newcommand{\Z}{\type{Z}}
\newcommand{\X}{\type{X}}
\newcommand{\Y}{\type{Y}}
\newcommand{\I}{\type{I}}

\newcommand{\E}{\type{E}}

\newcommand{\blank}{\boxdot}

\newcommand{\floor}[1]

\usepackage[utf8]{inputenc}
\usepackage{newunicodechar}
\newunicodechar{→}{\ensuremath{\to}}
\newunicodechar{⊗}{\ensuremath{\otimes}}
\newunicodechar{×}{\ensuremath{\times}}
\newunicodechar{∩}{\ensuremath{\cap}}
\newunicodechar{⊤}{\ensuremath{\top}}
\newunicodechar{⊥}{\ensuremath{\bot}}

\usepackage{etoolbox} %
\newtoggle{comments}
\toggletrue{comments}

\iftoggle{comments}{
\newcommand{\anote}[1]{\textcolor[rgb]{0.7,0,0.3}{(Aarthi: #1)}}
\newcommand{\bnote}[1]{\textcolor{orange}{(Brad: #1)}}
\newcommand{\anew}[1]{\textcolor[rgb]{0.7,0,0.3}{#1}}
\newcommand{\rnr}[1]{\textcolor[rgb]{0.1,0.7,0.1}{(Robert: #1)}}
\newcommand{\knote}[1]{\textcolor{blue}{(Kartik: #1)}}
\newcommand{\todo}[1]{\textcolor{red}{#1}}

}{
  \newcommand{\anote}[1]{}
  \newcommand{\bnote}[1]{}
  \newcommand{\anew}[1]{}
  \newcommand{\rnr}[1]{}  
  \newcommand{\knote}[1]{}
  \newcommand{\todo}[1]{}
  
}

\title{Gottesman Types for Quantum Programs}
\author{Robert Rand
\institute{University of Chicago}
\email{rand@uchicago.edu}
\and
Aarthi Sundaram
\institute{Microsoft Quantum}
\email{aarthims@gmail.com}
\and
Kartik Singhal
\institute{University of Chicago}
\email{ks@cs.uchicago.edu}
\and
Brad Lackey
\institute{Microsoft Quantum}
\institute{University of Maryland}
\email{bclackey@umd.edu}
}

\begin{document}
\maketitle

\begin{abstract}
The Heisenberg representation of quantum operators provides a powerful technique for reasoning about quantum circuits, albeit those restricted to the common (non-universal) Clifford set $H$, $S$ and $CNOT$. The Gottesman-Knill theorem showed that we can use this representation to efficiently simulate Clifford circuits. We show that Gottesman's semantics for quantum programs can be treated as a type system, allowing us to efficiently characterize a common subset of quantum programs. We also show that it can be extended beyond the Clifford set to partially characterize a broad range of programs. We apply these types to reason about separable states and the superdense coding algorithm.
\end{abstract}

\section{Introduction}

The \emph{Heisenberg representation} of quantum mechanics treats quantum operators as functions on unitary matrices, rather than on the quantum state. For instance, for any quantum state $\ket{\phi}$,
\begin{equation} HZ\ket{\phi} = XH\ket{\phi} \end{equation}
so an $H$ gate can be viewed as a higher-order function that takes $Z$ to $X$ (and similarly takes $X$ to $Z$).
Gottesman~\cite{Gottesman1998} uses this representation to present the rules for how the Clifford quantum gates $H$, $S$ and $CNOT$ operate on Pauli $X$ and $Z$ gates, which is sufficient to fully describe the behavior of the Clifford operators. There, H is given the following description based on its action above:
\[ H : \Z \to \X \qquad H : \X \to \Z \]

In Gottesman's paper, the end goal is to fully describe quantum programs and prove the \emph{Gottesman-Knill theorem}, which shows that any Clifford circuit can be simulated efficiently. Here we observe that the judgements above look like typing judgments, and show that they can indeed be treated as such (\cref{sec:basis-semantics}). As such, they can be used to make coarse guarantees about programs, without fully describing the programs' behaviors. We show a simple example of applying this system to the superdense coding algorithm (\cref{sec:superdense}). In \cref{section:GHZ}, we demonstrate, using the GHZ state $\ket{000} + \ket{111}$, how the type system is capable of tracking both the creation and destruction of entanglement. In \cref{sec:beyond-clifford}, we extend the type system to deal with programs outside the Clifford group and use it to characterize the Toffoli gate. We discuss the possible future applications of this system in \cref{sec:future}.

The system and examples in this paper are formalized in Coq at \url{https://github.com/inQWIRE/GottesmanTypes}.

\pagebreak

\section{Gottesman Semantics}
\label{sec:main}

Gottesman's semantics for $H$, $S$, and $CNOT$ are given by the following table.

\begin{figure}[h]
\[ \begin{array}{rlcrl}
H: & \X \rightarrow \Z \qquad & & CNOT: & \X \otimes \I \rightarrow \X \otimes \X\\
H: & \Z \rightarrow \X & & CNOT: & \I \otimes \X \rightarrow \I \otimes \X \\
S: & \X \rightarrow \Y & & CNOT: & \Z \otimes \I \rightarrow \Z \otimes \I \\
S: & \Z \rightarrow \Z & & CNOT: & \I \otimes \Z \rightarrow \Z \otimes \Z
\end{array} \]
\label{Basic }
\end{figure}

Note that these rules are intended to simply describe the action of each gate on the corresponding unitary matrices, as in equation 1. However, given the action of a circuit on every permutation of $X$ and $Z$ (or any spanning set) we can deduce the semantics of the program itself~\cite{Gottesman1998}.\footnote{The intuition behind this comes from the fact that this set of $X$ and $Z$ operators form the generators of the Pauli basis for quantum states. Hence, deducing the action of a unitary matrix on them provides an information theoretically complete picture of the action of the matrix on any input and suffices to deduce the semantics of the program.} For our purposes, though, we will treat these as a type system, justifying that choice in \cref{sec:basis-semantics}.

We can combine our typing rules by simple multiplication, for instance combining the second and third rules for $CNOT$, we get
\[ CNOT : (\I\Z \otimes \X\I \rightarrow \I\Z \otimes \X\I) = \Z \otimes \X \rightarrow \Z \otimes \X \]

In the rule for $S$, $\Y$ is equivalent to $i\X\Z$ so we can reason about an $S$ applied twice compositionally. We use $f \comp g$ for forward function composition (equivalent to $g \circ f$):
\[ S ; S : (\X \to i\X\Z \comp i\X\Z \to i\Y\Z) = \X \to -\X \]

Once we have a type system for $H$, $S$ and $CNOT$ we can define additional gates in terms of these, and derive their types. For instance, the Pauli $Z$ gate is simply $S;S$, for which we derived the following type for $\X$ and trivially can derive the type for $\Z$:
\begin{align*}
Z: ~& \X \rightarrow -\X \\
Z: ~& \Z \rightarrow \Z
\end{align*}

Defining $X$ as $H;Z;H$, we can derive the type for the Pauli $X$ gate as:
\begin{align*}
    X = H ; Z ; H: ~& (\X \to \Z \comp \Z \to \Z \comp \Z \to \X)  = \X \rightarrow \X \\
    X = H ; Z ; H: ~& (\Z \rightarrow \X \comp \X \rightarrow -\X \comp -\X \rightarrow -\Z)  = \Z \rightarrow -\Z
\end{align*}

Likewise, $Y = S; X; Z; S$ and so, the type for the Pauli $Y$ gate would be:
\begin{align*}
    Y = S; Z ; X ; S: ~& (\Z \rightarrow \Z \comp \Z \rightarrow \Z \comp \Z \rightarrow -\Z \comp -\Z \rightarrow - \Z) = \Z \rightarrow -\Z \\
    Y = S; Z ; X ; S: ~& (\X \rightarrow \Y \comp \Y \rightarrow -\Y \comp -\Y \rightarrow \Y \comp \Y \rightarrow -\X) = \X \rightarrow -\X
\end{align*}
We can also define more complicated gates like $CZ$ and $SWAP$ as $H_2; CNOT; H_2$ (where $H_2$ is $H$ applid to $CNOT$'s target qubit)
and $CNOT; NOTC; CNOT$, (where $NOTC$ is a flipped $CNOT$) for which we can easily derive the following types:
\[ \begin{array}{rlrl}
CZ: & \X \otimes \I \rightarrow \X \otimes \Z & \quad
SWAP: & \X \otimes \I \rightarrow \I \otimes \X \\
CZ: & \I \otimes \X \rightarrow \Z \otimes \X & \quad
SWAP: & \I \otimes \X \rightarrow \X \otimes \I \\
CZ: & \Z \otimes \I \rightarrow \Z \otimes \I & \quad
SWAP: & \Z \otimes \I \rightarrow \I \otimes \Z \\
CZ: & \I \otimes \Z \rightarrow \I \otimes \Z & \quad
SWAP: & \I \otimes \Z \rightarrow \Z \otimes \I \\
\end{array} \]

By combining the rules stated above, we can also derive the action of $SWAP$ on $\X \otimes \Y$ as:
\[SWAP: (\X \otimes \Y = \X\I\I \otimes \I(i\X)\Z) \rightarrow (\I(i\X)\Z \otimes \X\I\I = \Y \otimes \X) \]
This gives us less information than the types for $SWAP$ above, but might be the type we intend for a given use of swapping: In particular, we might want to use a $SWAP$ to exchange qubits in the $X$ and $Y$ bases, as we will show.

We show the full rules for typing circuits in \cref{fig:types}, which we will reference throughout this paper.

\begin{figure}
    \centering

\begin{enumerate}
    \item Grammar:
    \begin{align*}
        G := \I \mid \X \mid \Z \mid - G \mid i G \mid G * G \mid G \otimes G \mid G \to G \mid G \cap G \mid G \times G \mid \top
    \end{align*}
    \item Multiplication and Tensor Laws:
    \begin{center}
    \renewcommand{\arraystretch}{1.2}
    \begin{tabular}{>{\hspace{\labelsep}}l @{\qquad} >{\hspace{\labelsep}}l @{\qquad} >{\hspace{\labelsep}}l @{\qquad} >{\hspace{\labelsep}}l}
        $\X * \X = \I$ & $\Z * \Z = \I$ & $\Z * \X = - \X * \Z$ & $A * \I = A = \I * A$ \\
         $ - - A = A$ & $i (i A) = - A$ & $i (-A) = - (iA)$ & $A * (B * C)  = A * B * C$ \\
        \multicolumn{2}{c}{\hspace{\labelsep}$-A * B = -(A * B) = A * -B$} & \multicolumn{2}{c}{\hspace{\labelsep}$i A * B = i (A * B) = A * i B$}\\
        \multicolumn{2}{c}{\hspace{\labelsep}$A \otimes (B \otimes C) = (A \otimes B) \otimes C$} & \multicolumn{2}{c}{\hspace{\labelsep}$A \otimes B = (A \otimes \I) * (\I \otimes B)$}\\
        \multicolumn{2}{c}{\hspace{\labelsep}$iA \otimes B = i(A \otimes B) = A \otimes iB$} & \multicolumn{2}{c}{\hspace{\labelsep}$ - A \otimes B = - (A \otimes B) = A \otimes - B$}
    \end{tabular}
    \[
     \ (A \otimes B) * (C \otimes D) = (A*C) \otimes (B*D) \text{ where } \abs{A} = \abs{C}
    \]
    \end{center}

    \item Tensors Rules:
    \begin{center}
    \renewcommand{\arraystretch}{2}
    \begin{tabular}{c @{\qquad \qquad} c}
    $\inferrule*[Right=$\otimes_1$]{g \ 1 : A \rightarrow A' \and \abs{E} = m-1}{g \ m : E \otimes A \otimes E' \rightarrow E \otimes A' \otimes E'}$ &
    $\inferrule*[Right=$\otimes$-rev]{g \ 2 \ 1 : A \otimes B \rightarrow A' \otimes B'}{g \ 1 \ 2 : B \otimes A \rightarrow B' \otimes A'}$ \\
    \multicolumn{2}{c}{$\inferrule*[Right=$\otimes_2$]{g \ 1 \ 2 : A \otimes B \rightarrow A' \otimes B' \and
    \abs{E} = m-1 \and \abs{E'} = n - m - 1}{g \ m \ n : E \otimes A \otimes E' \otimes B \otimes E'' \rightarrow E \otimes A' \otimes E' \otimes B' \otimes E''}$}
    \end{tabular}
    \end{center}
    \item Arrow and Sequence Rules:
    \begin{center}
    \renewcommand{\arraystretch}{2}
    \begin{tabular}{c @{\qquad \qquad} c}
    $\inferrule*[Right=mul]{p : A \rightarrow A' \and p : B \rightarrow B'}{p : (A*B) \rightarrow (A'*B')}$ &
    $\inferrule*[Right=im]{p : A \rightarrow A'}{p : i A \rightarrow i A'}$\\
    $\inferrule*[Right=cut]{p_1 : A \rightarrow B \and p_2 : B \rightarrow C}{p_1 ; p_2 : A \rightarrow C}$ &
    $\inferrule*[Right=neg]{p : A \rightarrow A'}{p : - A \rightarrow - A'}$\\
    \multicolumn{2}{c}{$p_1 ; (p_2 ; p_3) : A \equiv (p_1 ; p_2) ; p_3 : A$}
    \end{tabular}
    \end{center}
    \item Intersection Rules:
    \[
    I^{\otimes n} \cap A = A \qquad A \cap A = A \qquad A \cap B = B \cap A
    \qquad A \cap B \cap C = A \cap (B \cap C)
    \]
    \renewcommand{\arraystretch}{2}
    \begin{center}
    \begin{tabular}{c @{\qquad \qquad} c}
    $\inferrule*[Right=$\cap$-I]{g : A \and g : B}{g : A \cap B}$
    &$\inferrule*[Right=$\cap$-E]{ g : A \cap B}{g : A}$
\\
    $\inferrule*[Right=$\cap$-Arr]{g : (A \rightarrow B) \cap (A \rightarrow C)}{g : A \rightarrow (B \cap C)}$ &
    $\inferrule*[Right=$\cap$-Arr-Dist]{g : (A \rightarrow A') \cap (B \rightarrow B')}{g : (A \cap B) \rightarrow (A' \cap B')}$
    \end{tabular}
    \end{center}
    \item Separability Rules
    \[
    (A \otimes \I^n) = A \times \I^n \text{ where } A \in \{1,-1,i,-i\} * \{X, Y, Z\}
    \]
    \[
    (A \times B) \cap (A \otimes C) = A \times (B \cap C)
    \qquad
    (A \times B) \cap (I^{\otimes n} \otimes C) = A \times (B \cap C)
    \]
\end{enumerate}

    \caption{The grammar and typing rules for Gottesman types. The grammar allows us to describe ill-formed types, such as $\X \cap (\I \otimes \Z)$, but these don't type any circuits. The intersection and arrow typing rules are derived from standard subtyping rules~\cite[Chapter 15]{Pierce:TypeSystems}. The arity of a type not containing $\{\times,\to\}$ is the longest sequence of atoms connected by tensors. For instance, $\abs{\I*\Z \otimes i\X} = \abs{(\X \otimes \X) \cap (\Z \otimes \Z)} = 2$.}
    \label{fig:types}
\end{figure}

\section{Interpretation on Basis States}
\label{sec:basis-semantics}

We can interpret the type $H : \Z \to \X$ as saying that $H$ takes a qubit in the $\Z$ basis (that is, $\ket{0}$ or $\ket{1}$) to a qubit in the \X basis ($\ket{+}$ and $\ket{-}$). This form of reasoning is present in Perdrix's \cite{Perdrix2008} work on abstract interpretation for quantum systems, which classifies qubits in an \X or \Z basis, for the purpose of tracking entanglement. Unfortunately, that system cannot leave the \X and \Z bases, and hence cannot derive that $Z : \X \rightarrow -\X$ due to the intermediate \Y. More fundamentally, while we can show that $SWAP$ has type $\X \otimes \Z \rightarrow \Z \otimes \X$, the first $CNOT$ application entangles the two qubits (represented in our system by $\X\Z \otimes \X\Z$) which Perdrix classifies as simply $\top$ and marks as potentially entangled. (We also use $\top$ but only for gates that lie outside the set of Clifford operators, see \cref{sec:beyond-clifford}.)
As a result, Perdrix's system typically classifies most circuits as $\top$ after just a few gate applications, while ours never leaves the Pauli bases as long as we apply Clifford gates.

\begin{proposition}
\label{prop:eigensate}
    Given a unitary $U : A \rightarrow B$ in the Heisenberg interpretation, $U$ takes every eigenstate of $A$ to an eigenstate of $B$.
\end{proposition}

\begin{proof}
    From eq. [1] in Gottesman~\cite{Gottesman1998}, given a state $\ket{\psi}$ and an operator $U$,
    \[ U N \ket{\psi} = U N U^{\dagger} U \ket{\psi}. \]
    Additionally, in the Heisenberg interpretation this can be denoted as: $U : N \rightarrow U N U^{\dagger}$. Suppose that $\ket{\psi}$ is an eigenstate of $N$ with eigenvalue $\lambda$ and let $\ket{\phi}$ denote the state after $U$ acts on $\ket{\psi}$. Then,
    \[ \lambda \ket{\phi} = U (\lambda \ket{\psi}) = U N \ket{\psi} = U N U^{\dagger} U \ket{\psi} = (U N U^{\dagger}) \ket{\phi}.
    \]

    Hence, $\ket{\phi}$ is an eigenstate of the modified operator $U N U^{\dagger}$ with eigenvalue $\lambda$.
\end{proof}

\paragraph{Intersection Types}

It may seem odd that we have given multiple types to $H$, $S$, $CNOT$ and the various derived operators. This isn't particularly rare in type systems with subtyping, which is an appropriate lens for viewing the types we have given above. However, it is useful to have the most descriptive type for any term in our language. We can give these using intersection types. For instance, we have the following types for $H$ and $CNOT$:
\begin{align*}
H :  &(\X \to \Z) \cap (\Z \to \X) \\
CNOT :  &(\X \otimes \I \to \X \otimes \X) \cap (\I \otimes \X \to \I \otimes \X) \cap (\Z \otimes \I \to \Z \otimes \I) \cap (\I \otimes \Z \to \Z \otimes \Z)
\end{align*}

From these we can derive any of the types given earlier, using the standard rules for intersections (\cref{fig:types}). Another advantage of using intersection types is that it is closely related to the stabilizer formalism that is used extensively in error-correction. This connection is further discussed in~\cref{sec:future}.

\paragraph{The role of \I}

\I plays an interesting role in this type system. For \I alone, we have the following two facts, the first drawn from the Heisenberg representation of quantum mechanics and the second from our interpretation of types as eigenstates:
\begin{align}
    &\forall~U, \quad U : \I \to \I  \\
    &\forall~\ket{\psi},~\ket{\psi} : \I
\end{align}

This would lead us to treat \I as a kind of top type, where $A <: \I$ for any type $A$. However, this would be incompatible with $\otimes$. For example, the two qubit Bell pair $\ket{\Phi^+}$ has type $\X \otimes \X$ but not type $\X \otimes \I$. By contrast, $\X \otimes \I$ contains precisely the separable two qubit states where the first qubit is an eigenstate of $\X$. This second type, which is neither a subtype nor supertype of the first allows us to consider the important property of \emph{separability} or non-entanglement of qubits.

\section{Separability}
\label{sec:separability}

\begin{proposition}\label{prop:separability}
     For any Pauli matrix $U \in \{-1, 1, -i, i\} * \{ X,  Y,  Z\}$, the eigenstates of $U \otimes I^{\otimes n-1}$ are all the vectors $\ket{u} \otimes \ket{\psi}$ where $\ket{u}$ is an eigenstate of $U$ and $\ket{\psi}$ is any state.
\end{proposition}

\begin{proof}
    Let $\ket{\phi}$ be the $\lambda$ eigenstate and $\ket{\phi^{\bot}}$ be the $-\lambda$ eigenstate for $U \in \{1, -1, i, -i\} * \{ X,  Y,  Z\}$ where $\lambda \in \{1, i\}$. Note that $\{\ket{\phi}, \ket{\phi^{\bot}}\}$ forms a single-qubit basis.

    First, consider states of the form $\ket{\gamma} = \ket{u} \otimes \ket{\psi}$ where $\ket{u} \in \{\ket{\phi}, \ket{\phi^{\bot}}\}$ and $\ket{\psi} \in \mathbb{C}^{2^{n-1}}$. Clearly,
    \[(U \otimes I^{\otimes n-1}) \ket{\gamma} = (U \otimes I^{\otimes n-1}) \ket{u} \otimes \ket{\psi} = (U \ket{u} ) \otimes \ket{\psi} = \lambda_u \ket{u} \otimes \ket{\psi}. \]
    Hence, every state of the form of $\ket{\gamma}$ is an eigenstate of $U \otimes I^{\otimes n-1}$. Additionally, note that by similar reasoning, for every separable state $\ket{\gamma} = \ket{v} \otimes \ket{\psi}$ where $\ket{v} \notin \{\ket{\phi}, \ket{\phi^{\bot}}\}$, $\ket{\gamma}$ is not an eigenstate of $U \otimes I^{\otimes n-1}$.

    Now we show, that any state not in this separable form cannot be an eigenstate of $U \otimes I^{\otimes n-1}$. By way of contradiction assume that $\ket{\delta}$ is an eigenstate of $U \otimes I^{\otimes n-1}$ with $(U \otimes I^{\otimes n-1}) \ket{\delta} = \mu \ket{\delta}$. Expand
    $$\ket{\delta} = \alpha \ket{\phi}\ket{\psi_1} + \beta\ket{\phi^{\bot}}\ket{\psi_2}$$
    where $\ket{\psi_1},\ket{\psi_2} \in \mathbb{C}^{2^{n-1}}$. Then we compute
    \begin{eqnarray*}
    (U \otimes I^{\otimes n-1}) \ket{\delta} & = & \alpha(U \ket{\phi}) \ket{\psi_1} +  \beta(U \ket{\phi^{\bot}}) \ket{\psi_2} \\
    & = & \lambda \alpha\ket{\phi}\ket{\psi_1} - \lambda \beta\ket{\phi^{\bot}}\ket{\psi_2}\\
    & = & \mu \alpha \ket{\phi} \ket{\psi_1} + \mu \beta \ket{\phi^{\bot}} \ket{\psi_2}
    \end{eqnarray*}
    where we have used that $\ket{\phi}$ and $\ket{\phi^{\bot}}$ are the $+\lambda$ and $-\lambda$ eigenvalues of $U$ respectively. As the components of the expansion are orthogonal to each other, $\mu$ must satisfy:
    \[
    \mu \alpha = \lambda \alpha \text{ and } \mu \beta = - \lambda \beta.
    \]
    As $\lambda \not= 0$, since $U \otimes I^{\otimes n-1}$ is unitary, we either have (i) $\alpha = 0$, $\mu = - \lambda$, and $\ket{\delta} = \ket{\phi^{\bot}}\ket{\psi_2}$ or (ii) $\beta = 0$, $\mu = +\lambda$, and $\ket{\delta} = \ket{\phi}\ket{\psi_1}$. In either case $\ket{\delta}$ has a separable form as claimed.
\end{proof}

Following Gottesman's notation, let $U_k := I^{\otimes (n-k)} \otimes U \otimes I^{\otimes (n-k)}$ for $1 \leq k \leq n$ and a single qubit Pauli $U$. (Gottesman writes $\overline{\X}_k$ and $\overline{\Z}_k$; we'll elide the bar but make frequent use of the notation.)
Combining~\Cref{prop:eigensate,prop:separability}, we obtain the following corollary:

\begin{corollary}
Every term of type $U_k$ is separable, for $U \in \{\pm X, \pm Y, \pm Z\}$ and $1 \leq k \leq n$. That is, the factor associated with $U$ is unentangled with the rest of the system.
\end{corollary}

In this light, we can reconsider the types we've given to $CNOT$. In \cref{sec:main}, we showed that $CNOT$ has the type $\Z \otimes \X \to \Z \otimes \X$. This is true but not particularly useful to a programmer who doesn't view the eigenvectors of $\Z \otimes \X$ as a helpful category. If however, we look at the intersection type given for $CNOT$ in \cref{sec:basis-semantics}, we see that $CNOT$ has the type $(\Z \otimes \I \to \Z \otimes \I) \cap (\I \otimes \X \to \I \otimes \X)$. This says that $CNOT$ takes $\ket{i} \otimes \ket{\psi}$ where $i\in \{0,1\}$ to a similar separable state, and likewise for when the second qubit is in \X. Hence, we can conclude that it takes a separable pair of \Z and \X qubits to a similar separable pair.

Our $\cap$-\textsc{Arr-Dist} rule (\cref{fig:types}), which follows directly from the subtyping rules for arrow and intersection\footnote{Thanks to Andreas Rossberg for \href{https://stackoverflow.com/a/61266762/1167061}{pointing this out on Stack Overflow}.}, allows us to make this explicit. We can weaken $CNOT$'s type to $(\Z \otimes \I \cap \I \otimes \X) \to (\Z \otimes \I \cap \I \otimes \X)$.
This type, which encodes the separability of the \Z and \X qubits, is useful enough that we will introduce a new type constructor $\times$, where $A \times \I^{\otimes n} = A \otimes \I^{\otimes n}$ for $A \in \{\pm X, \pm Y, \pm Z\}$. This distributes in the expected way over intersections: $(A \times B) \cap (A \otimes C) = A \times (B \cap C)$ and $(A \times B) \cap (I \otimes C) = A \times (B \cap C)$, allowing us to derive $CNOT : \Z \times \X$.

\section{Example: Superdense Coding}
\label{sec:superdense}

To illustrate the power of this simple system, consider the example of superdense coding as in \cref{fig:superdense}. Superdense coding allows Alice to convey two bits of information $x$ and $y$, which we treat as qubits in the $\Z$ state, to Bob by sending a single qubit and consuming one EPR pair.

\begin{figure}[ht]
\begin{center} \begin{tikzpicture}[scale=0.7,y=-1cm]

    \draw[double] (-0.4,0) node[left] {$x$} -- (10.5,0);
    \draw (-0.4,1) node[left] {$y$} edge[double] (10.5,1);
    \draw (-0.4,2) node[left] {$\ket{0}$} -- (7,2); \draw (7,2) -- (10.5,2) node[right] {$x$};
    \draw (-0.4,3) node[left] {$\ket{0}$} -- (10.5,3) node[right] {$y$};

	\unitary{H}{0.5}{2}
	\ctrl{2}{2}
    \draw (2,2) -- (2,3);
	\unitary{X}{2}{3}

    \draw[style=double] (4,0) -- (4,2);
	\ctrl{4}{0}
	\unitary{Z}{4}{2}
    \draw[style=double] (5.5,1) -- (5.5,2);
	\ctrl{5.5}{1}
	\unitary{X}{5.5}{2}

    \draw (8,2) -- (8,3);
	\ctrl{8}{2}
	\unitary{X}{8}{3}
	\unitary{H}{9.5}{2}

\end{tikzpicture} \end{center}
\caption{Superdense Coding sending classical bits $x$ and $y$ from Alice to Bob.}
\label{fig:superdense}
\end{figure}
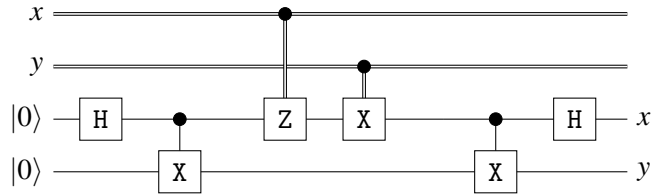

The desired type for this circuit is $\Z \times \Z \times \Z \times \Z$, showing that four classical bits are transmitted. Hence, we need to derive the output types for
$\Z_1, \Z_2, \Z_3$ and $\Z_4$. We can trivially derive that \coqe{superdense} has types $\Z_1 \to \Z_1$ and $\Z_2 \to \Z_2$ (since $CZ$ and $CNOT$ have types $\Z \otimes \I \to \Z \otimes \I$), so we'll look at the derivation for $\Z_3$:

\begin{coq}
Definition superdense :=
  INIT ;     I ⊗ I ⊗ Z ⊗ I   (* initial type *)
  H 3  ;     I ⊗ I ⊗ X ⊗ I   (* create Bell pair *)
  CNOT 3 4 ; I ⊗ I ⊗ X ⊗ X
  CZ 1 3 ;   Z ⊗ I ⊗ X ⊗ X   (* map bits onto Bell pair *)
  CNOT 2 3 ; Z ⊗ I ⊗ X ⊗ X
  CNOT 3 4 ; Z ⊗ I ⊗ X ⊗ I   (* decode qubits *)
  H 3        Z ⊗ I ⊗ Z ⊗ I
\end{coq}

We can similarly derive \coqe{superdense : I ⊗ I ⊗ I ⊗ Z -> I ⊗ Z ⊗ I ⊗ Z}. We can now combine all four derivations using our distributivity rule for $\to$ and $\cap$ as follows:
\[
\infer
{\mathtt{superdense} : (\Z_1 \to \Z_1) \cap (\Z_2 \to \Z_2) \cap (\Z_3 \to \Z \otimes \I \otimes \Z \otimes \I) \cap (\Z_4 \to \I \otimes \Z \otimes \I \otimes \Z)}
{\infer
{\mathtt{superdense} : \Z_1 \cap \Z_2 \cap \Z_3 \cap \Z_4 \to \Z_1 \cap \Z_2 \cap \Z \otimes \I \otimes \Z \otimes \I \cap \I \otimes \Z \otimes \I \otimes \Z}
{\infer
{\mathtt{superdense} : \Z \times \Z \times \Z \times \Z \to \Z \times (\Z \otimes \I \otimes \I \cap \I \otimes \Z \otimes \I \cap \Z \otimes \I \otimes \Z)}
{\infer
{\mathtt{superdense} : \Z \times \Z \times \Z \times \Z \to \Z \times (\Z \times (\Z \otimes \I \cap \I \otimes \Z)}
{\mathtt{superdense} : \Z \times \Z \times \Z \times \Z \to \Z \times \Z \times \Z \times \Z}
}
}
}
\]

\section{Example: GHZ state, Creation and Unentangling}\label{section:GHZ}

To demonstrate how we can track the possibly entangling and disentangling property of the CNOT gate, we can look at the example of creating the GHZ state $\ket{000} + \ket{111}$ starting from $\ket{000}$ and then disentangling it. A similar example was considered by Honda~\cite{Honda2015} to demonstrate how his system can track when $CNOT$ displays either its entangling or un-entangling behaviour. One crucial difference is that Honda uses the denotational semantics of density matrices which, in practice, would scale badly with the size of the program being type checked.
Our approach is closer to that of Perdrix~\cite{Perdrix2007,Perdrix2008} in terms of design and scalability but capable of showing separability where the prior systems could not.

We will consider the following GHZ program acting on the initial state $\Z \times \Z \times \Z = \Z_1 \cap \Z_2 \cap \Z_3$. We first follow the derivation for $\Z_1$:

\begin{coq}
  Definition GHZ :=
    INIT ;      Z ⊗ I ⊗ I   (* initial state *)
    H 1;        X ⊗ I ⊗ I
    CNOT 1 2;   X ⊗ X ⊗ I   (* Bell Pair *)
    CNOT 2 3:   X ⊗ X ⊗ X   (* GHZ State created *)
\end{coq}

Repeating the derivation for $Z_2$ and $Z_3$, We obtain the following type:
\[
\mathtt{GHZ} :
(\Z_1 → \X ⊗ \X ⊗ \X) ∩
(\Z_2 → \Z ⊗ \Z ⊗ \I) ∩
(\Z_3 → \I ⊗ \Z ⊗ \Z)
\]
This type is non-trivial to read, but we can clearly see that entanglement is produced between the three qubits.

If we now apply \coqe{CNOT 2 1}, we get the following type:
\[
\mathtt{GHZ;~CNOT~2~1} :
(\Z_1 → \I ⊗ \X ⊗ \X) ∩
(\Z_2 → \Z ⊗ \I ⊗ \I) ∩
(\Z_3 → \I ⊗ \Z ⊗ \Z)
\]
We can reduce the output type to $\Z \times (\X \otimes \X \cap \Z \otimes \Z)$, showing that we've neatly separated the first qubit from the Bell pair on qubits 2 and 3 in what used to be a GHZ state.

If we then apply \coqe{CNOT 3 2} we get
\[
\mathtt{GHZ;~CNOT~2~1;~CNOT~3~2} :
(\Z_1 → \I ⊗ \I ⊗ \X) ∩
(\Z_2 → \Z ⊗ \I ⊗ \I) ∩
(\Z_3 → \I ⊗ \Z ⊗ \I)
\]
resulting in the fully separable state $\Z \times \Z \times \X$.

\section{Beyond the Clifford Group}
\label{sec:beyond-clifford}
Universal quantum computation requires that we use gates outside the Clifford set, the two most studied candidates being the $T$ ($\pi/8$) and Toffoli gates. In order to type $T$ and Toffoli, we will need to add a new $\top$ type, which describes any basis state in the state interpretation of our types, and any unitary in the Heisenberg interpretation. We can then give the following type to the $T$ gate:
\begin{align*}
T: ~& \Z \rightarrow \Z \\
T: ~& \X \rightarrow \top
\end{align*}

Here $\top$ really is a top type. Unlike $\I$, which behaves like an identity, $\top$ behaves like an annihilator:
\begin{align*}
    \forall A, \I A = A = A\I \\
    \forall A, \top A = \top = A\top
\end{align*}

Instead of explicitly giving a type to the Toffoli gate, we can derive it from Toffoli's standard decomposition into $T$, $H$ and $CNOT$ gates:
\begin{coq}
Definition TOFFOLI a b c :=
  H c;
  CNOT b c; T† c;
  CNOT a c; T c;
  CNOT b c; T† c;
  CNOT a c; T b; T c;
  H c;
  CNOT a b; T a; T† b;
  CNOT a b.
\end{coq}

We first note that $T^\dagger$ is simply seven consecutive $T$s, so like $T$, it preserves $\Z$ and takes $\X$ to $\top$.
Now consider \coqe{TOFFOLI}'s action on $\I \otimes \I \otimes \X$. We can annotate the program with the types after every line:
\begin{coq}
Definition TOFFOLI a b c :=
  INIT ;                I ⊗ I ⊗ X
  H c ;                 I ⊗ I ⊗ Z
  CNOT b c; T† c;       I ⊗ Z ⊗ Z
  CNOT a c; T c;        Z ⊗ Z ⊗ Z
  CNOT b c; T† c;       Z ⊗ I ⊗ Z
  CNOT a c; T b; T c;   I ⊗ I ⊗ Z
  H c;                  I ⊗ I ⊗ X
  CNOT a b; T a; T† b;  I ⊗ I ⊗ X
  CNOT a b.             I ⊗ I ⊗ X
\end{coq}
The $T$'s here are all identities, so only the behaviors of $H$ (which takes $\X$ to $\Z$ and back) and $CNOT$ (which takes $\I \otimes \Z$ to $\Z \otimes \Z$ and back) are relevant.

The derivations for $\Z \otimes \I \otimes \I$ and $\I \otimes \Z \otimes \I$ are similar, and the other three states map to the top element:

\[ \begin{array}{rlrl}
\mathtt{TOFFOLI}: & \Z \otimes \I \otimes \I \rightarrow \Z \otimes \I \otimes \I & \quad
\mathtt{TOFFOLI}: & \X \otimes \I \otimes \I \rightarrow \top \otimes \top \otimes \top \\
\mathtt{TOFFOLI}: & \I \otimes \Z \otimes \I \rightarrow \I \otimes \Z \otimes \I & \quad
\mathtt{TOFFOLI}: & \I \otimes \X \otimes \I \rightarrow \top \otimes \top \otimes \top \\
\mathtt{TOFFOLI}: & \I \otimes \I \otimes \Z \rightarrow \top \otimes \top \otimes \top & \quad
\mathtt{TOFFOLI}: & \I \otimes \I \otimes \X \rightarrow  \I \otimes \I \otimes \X \\
\end{array} \]

Using our technique for deriving judgements about separability, we can further derive the following type for the Toffoli gate
\[ \mathtt{TOFFOLI} : \Z \times \Z \times \X \rightarrow \Z \times \Z \times \X \]
which says that if you feed a Toffoli three separable qubits in the $\Z$, $\Z$ and $\X$ basis, you receive three qubits in the same basis back.

\section{Applications and Future Work}
\label{sec:future}

We think that this type system, along with possible extensions, is broadly applicable. Here we outline some of the possible uses of the type system along with (where necessary) extensions that will enable these uses.

\paragraph{Stabilizer Types and Quantum Error Correction}

This aim of this paper is to define types for unitary operations, however in \cref{sec:basis-semantics} we interpreted types such as $\X\otimes\X$ and $\X\otimes\I$ as being inhabited by certain states. Gates and circuits, by acting on such states, obtain an arrow type. But as we vary the input states on which unitaries act, they in fact obtain many different arrow types. We used the concept of intersection types to characterize this phenomenon.

However the nature of Pauli operators allows for a different treatment using the stabilizer formalism~\cite{Gottesman1998}. Namely two tensor products of Pauli operators (of the same arity) must either commute or anticommute. Consequently, any collection of such types with a joint eigenspace (and hence having a nonempty intersection type) necessarily pairwise commutes. Then, in fact, any product of these operators will also share this eigenspace and would serve equally well in describing the intersection type.

For example, $X\otimes X$ and $X\otimes I$ commute and the intersection type $(\X\otimes\X) \cap (\X\otimes\I)$ is nonempty, being $\{\ket{++},\ket{+-},\ket{-+},\ket{--}\}$. These states are also eigenstates of $I\otimes X = (X \otimes X)(X\otimes I)$ and hence also of type $\I\otimes\X = (\X\otimes\X) * (\X\otimes\I)$. Consequently our intersection type could be equally well presented as $(\X\otimes\X) \cap (\I\otimes\X)$ or $(\I\otimes\X) \cap (\X\otimes\I)$. And so we should not think of the intersection type $(\X\otimes\X) \cap (\X\otimes\I)$ being determined by $\X\otimes\X$ and $\X\otimes\I$ alone, but rather by all the elements in the (commutative) group they generate.

Such a commutative group is typically called a stabilizer group. Formally, a stabilizer group is any commutative group of tensor products of Pauli operators some common arity $n$ that does not contain $-I^{\otimes n}$. This latter condition ensures the joint $+1$-eigenspace of all the operators in a stabilizer group exists, which then is called the stabilizer code associated to the group~\cite{gottesman1996class}. The eigenspaces associated with other eigenvalue combinations are called syndrome spaces. Therefore our intersection type (presuming it does not include $-\I^n$) is realized as a stabilizer code and its associated syndrome spaces, hence capturing the notion of a Pauli frame~\cite{knill2005quantum}.

The Gottesman semantics of \cref{sec:main} then can be interpreted as type theory for stabilizer groups and stabilizer codes, which is the direction taken by \cite{Honda2015}. This provides a type theory for stabilizer codes that includes encoding and decoding~\cite{cleve1997efficient}, and syndrome extraction~\cite{steane1997active} circuits. Potential other applications could include analyzing circuits that implement non-Clifford gates on stabilizer codes~\cite{paetznick2013universal, yoder2017universal} and circuits that fault-tolerantly switch between codes~\cite{colladay2018rewiring}.

Beyond application to quantum error correction, treating intersection types through stabilizer groups offers new deduction rules. In the analysis of the GHZ state of $\cref{section:GHZ}$, the type of \texttt{GHZ; CNOT 1 3} has as codomain the intersection type
$$(\X\otimes \X\otimes \I) \cap (\Z \otimes \Z \otimes \I) \cap (\Z \otimes \Z \otimes \Z).$$
Given the deduction rules of \cref{fig:types} alone, proving separability of this type is impossible. Yet through the stabilizer formalism, one recognizes it also has the type $(\I \otimes \I \otimes \Z) = (\Z \otimes \Z \otimes \I) * (\Z \otimes \Z \otimes \Z)$, and so we may represent our intersection type equivalently as
$$(\X\otimes \X\otimes \I) \cap (\Z \otimes \Z \otimes \I) \cap (\I \otimes \I \otimes \Z) = ((\X\otimes \X) \cap (\Z\otimes \Z))\times \Z.$$

\paragraph{Adding measurement}

It's challenging to turn Gottesman's semantics for measurement into a type system in light of the fact that it looks at the operation on all the basis states, rather than simply the evolution of a single Pauli operator. Namely it adds significant computational complexity, where typechecking should be linear. Nonetheless, the action of measurement on stabilizer groups is understood~\cite{Gottesman1998}. With the link between intersection types and stabilizers established above, we can see the potential for realizing such a type system. The challenge lies in the fact that to execute the action of measurement on a stabilizer group, one needs to find a set of generators of the group in a particular form. In our grammar, we need to express an intersection type using a special presentation, which requires the introduction of rewrite rules like those we just used in the analysis of the GHZ state.

Formally measurement acts as follows; for ease of exposition, suppose we are measuring the first qubit in the computational basis (a $Z$-basis measurement). From our initial intersection, do the following:
\begin{enumerate}
    \item Use the rewrite rules to ensure there are only $0$ or $1$ terms in the intersection that have $\X$ in the first position. If there is $1$ such term, then remove it from the intersection.
    \item If there are $0$ terms that have $\X$ in the first position, use the rewrite rules to ensure there are only $0$ or $1$ terms that have $\Z$ in the first position. If there is $1$ such term, remove it from the intersection.
    \item Add $\Z\otimes\I^{n-1}$ to the intersection.
\end{enumerate}
The resulting intersection type is the type after measurement.

For example in our analysis of the GHZ state in \cref{section:GHZ}, the circuit \texttt{GHZ} had a codomain of type
$$(\X\otimes \X\otimes \X) \cap (\Z \otimes \Z \otimes \I) \cap (\I \otimes \Z \otimes \Z).$$
To compute the type of \texttt{GHZ; MEAS 1} we enact the above program. Fortunately our intersection already has the requisite form, with the first term being the only one with an $\X$ in the initial position. We remove the term and add $\Z\otimes\I\otimes\I$ to get
$$(\Z\otimes \I\otimes \I) \cap (\Z \otimes \Z \otimes \I) \cap (\I \otimes \Z \otimes \Z).$$
Using these same rewrite rules, we can replace the second term with $\I \otimes \Z \otimes \I$ and with that term replace the last with $\I \otimes \I \otimes \Z$, producing the output type as
$$(\Z\otimes \I\otimes \I) \cap (\Z \otimes \I \otimes \I) \cap (\I \otimes \I \otimes \Z) = \Z \times \Z \times \Z.$$

At first glance, it is not clear that the program for measurement above is well-defined. Why should such a presentation for a general intersection type exist? Even if one does, how do we find one among all possible rewrites? Fortunately neither of the concerns is an issue. It is well known that stabilizer groups have a denotational semantics as binary matrices of dimension linear in the number of qubits, see for example \cite[\S{10.5}]{Nielsen2010}. The existence of such a presentation, and a procedure to find one, reduces to straightforward linear algebra. Hence the program above for measurement can be accomplished in at worst quadratic time in the number of qubits.

\vspace{-0.5em}
\paragraph{Resource Tracking}
Resource monotones track the amount of resources contained in a type. For instance, at a very coarse level, one might quantify the entanglement in a state by counting the number of ``ebits'' needed to create the state. In our formalism, a Bell state is of type $\E = (\X\otimes \X) \cap (\Z\otimes \Z)$, which is counted as having one ebit.  A seperable type such as $\X \times \Z$ has a resource value 0 and hence no ebits. While the current system cannot handle such resource monotones, it seems plausible that using a suitable extension could similarly calculate the resource cost of various operations. Say, by counting the number of $\E$ types used in a protocol. Then, superdense coding has an $\E$-count of 1 while entanglement swapping will have an $\E$-count of 2.

Our types are too fine to capture the notion of local equivalence~\cite{van2005local}. Instead, one can coarsen to types generated from graph states~\cite{fattal2004entanglement}. For example, all six entangled two qubit types are equivalent under the local operations of $H\otimes I$, $S\otimes I$, $I\otimes H$, and $I\otimes S$, and therefore, from an entanglement monotone perspective, they all contain the same amount of entanglement.

One can show that for three qubits, there are only two classes of entangled types, up to relabeling of the qubit indices~\cite{bravyi2006ghz}. Such computations are quite challenging at scale, and so methods to quantify the amount of entanglement using such states is a long-term goal of this project.

\vspace{-0.5em}
\paragraph{Ancilla correctness.}
Many quantum circuits introduce ancillary qubits that are used to perform some classical computation and are then discarded in a basis state.
Several efforts have been made to verify this behaviour: The Quipper~\cite{Green2013} and Q\#~\cite{Svore2018} languages allow us to \emph{assert} that ancilla are separable and can be safely discarded, while \qwire allows us to manually verify this~\cite{Rand2018}. More recently, Silq~\cite{Silq} allows us to define ``qfree'' functions that never put qubits into a superposition. We hope to avoid this restriction and use our type system to automatically guarantee ancilla correctness by showing that the ancillae are in $\Z$ and separable.

\vspace{-0.5em}
\paragraph{Provenance tracking}
Another useful addition to the type system is \emph{ownership}. Superdense coding is a central example of a class of quantum communication protocols. By annotating the typing judgements with ownership information and restricting multiqubit operations to qubits under the same ownership, we can guarantee that superdense coding only transmits a single qubit, via a provided channel $\mathcal{C}$.
(Note that measurement and ownership types are both forms of static information-flow control~\cite{Sabelfeld2003}.)
With this additional typing information, we can give superdense coding the type
\[ \Z_A \times\Z_A\times\Z_A\times\Z_B \to \Z_A\times\Z_A\times\Z_B\times\Z_B \]
indicating that a single qubit has passed from Alice's control to Bob's.

\subsection*{Acknowledgements}
The first author would like to acknowledge the support of the U.S. Department of Energy, Office of Science, Office of
Advanced Scientific Computing Research, Quantum Testbed Pathfinder Program under Award Number DE-SC0019040.
The third author was funded by EPiQC, an NSF Expedition in Computing, under grant CCF-1730449.

\bibliographystyle{eptcs}
\bibliography{references}

\begin{thebibliography}{10}
\providecommand{\bibitemdeclare}[2]{}
\providecommand{\surnamestart}{}
\providecommand{\surnameend}{}
\providecommand{\urlprefix}{Available at }
\providecommand{\url}[1]{\texttt{#1}}
\providecommand{\href}[2]{\texttt{#2}}
\providecommand{\urlalt}[2]{\href{#1}{#2}}
\providecommand{\doi}[1]{doi:\urlalt{http://dx.doi.org/#1}{#1}}
\providecommand{\bibinfo}[2]{#2}

\bibitemdeclare{inproceedings}{Silq}
\bibitem{Silq}
\bibinfo{author}{Benjamin \surnamestart Bichsel\surnameend},
  \bibinfo{author}{Maximilian \surnamestart Baader\surnameend},
  \bibinfo{author}{Timon \surnamestart Gehr\surnameend} \&
  \bibinfo{author}{Martin \surnamestart Vechev\surnameend}
  (\bibinfo{year}{2020}): \emph{\bibinfo{title}{Silq: A High-Level Quantum
  Language with Safe Uncomputation and Intuitive Semantics}}.
\newblock In: {\sl \bibinfo{booktitle}{Proc. PLDI '20}}, p.
  \bibinfo{pages}{286–300}, \doi{10.1145/3385412.3386007}.
\newblock
  \urlprefix\url{https://files.sri.inf.ethz.ch/website/papers/pldi20-silq.pdf}.

\bibitemdeclare{article}{bravyi2006ghz}
\bibitem{bravyi2006ghz}
\bibinfo{author}{Sergey \surnamestart Bravyi\surnameend},
  \bibinfo{author}{David \surnamestart Fattal\surnameend} \&
  \bibinfo{author}{Daniel \surnamestart Gottesman\surnameend}
  (\bibinfo{year}{2006}): \emph{\bibinfo{title}{GHZ extraction yield for
  multipartite stabilizer states}}.
\newblock {\sl \bibinfo{journal}{J. Math. Phys.}}
  \bibinfo{volume}{47}(\bibinfo{number}{6}), p. \bibinfo{pages}{062106},
  \doi{10.1063/1.2203431}.
\newblock \urlprefix\url{https://arxiv.org/abs/quant-ph/0504208}.

\bibitemdeclare{article}{cleve1997efficient}
\bibitem{cleve1997efficient}
\bibinfo{author}{Richard \surnamestart Cleve\surnameend} \&
  \bibinfo{author}{Daniel \surnamestart Gottesman\surnameend}
  (\bibinfo{year}{1997}): \emph{\bibinfo{title}{Efficient Computations of
  Encodings for Quantum Error Correction}}.
\newblock {\sl \bibinfo{journal}{Phys. Rev. A}} \bibinfo{volume}{56}, pp.
  \bibinfo{pages}{76--82}, \doi{10.1103/PhysRevA.56.76}.
\newblock \urlprefix\url{https://arxiv.org/abs/quant-ph/9607030}.

\bibitemdeclare{article}{colladay2018rewiring}
\bibitem{colladay2018rewiring}
\bibinfo{author}{Kristina~R \surnamestart Colladay\surnameend} \&
  \bibinfo{author}{Erich~J \surnamestart Mueller\surnameend}
  (\bibinfo{year}{2018}): \emph{\bibinfo{title}{Rewiring Stabilizer Codes}}.
\newblock {\sl \bibinfo{journal}{New Journal of Physics}}
  \bibinfo{volume}{20}(\bibinfo{number}{8}), p. \bibinfo{pages}{083030},
  \doi{10.1088/1367-2630/aad8dd}.
\newblock \urlprefix\url{https://arxiv.org/abs/1707.09403}.

\bibitemdeclare{misc}{fattal2004entanglement}
\bibitem{fattal2004entanglement}
\bibinfo{author}{David \surnamestart Fattal\surnameend},
  \bibinfo{author}{Toby~S \surnamestart Cubitt\surnameend},
  \bibinfo{author}{Yoshihisa \surnamestart Yamamoto\surnameend},
  \bibinfo{author}{Sergey \surnamestart Bravyi\surnameend} \&
  \bibinfo{author}{Isaac~L \surnamestart Chuang\surnameend}
  (\bibinfo{year}{2004}): \emph{\bibinfo{title}{Entanglement in the stabilizer
  formalism}}.
\newblock \bibinfo{howpublished}{arXiv preprint}.
\newblock \urlprefix\url{https://arxiv.org/abs/quant-ph/0406168}.

\bibitemdeclare{article}{gottesman1996class}
\bibitem{gottesman1996class}
\bibinfo{author}{Daniel \surnamestart Gottesman\surnameend}
  (\bibinfo{year}{1996}): \emph{\bibinfo{title}{Class of quantum
  error-correcting codes saturating the quantum Hamming bound}}.
\newblock {\sl \bibinfo{journal}{Phys. Rev. A}}
  \bibinfo{volume}{54}(\bibinfo{number}{3}), p. \bibinfo{pages}{1862–1868},
  \doi{10.1103/physreva.54.1862}.
\newblock \urlprefix\url{https://arxiv.org/abs/quant-ph/9604038}.

\bibitemdeclare{inproceedings}{Gottesman1998}
\bibitem{Gottesman1998}
\bibinfo{author}{Daniel \surnamestart Gottesman\surnameend}
  (\bibinfo{year}{1998}): \emph{\bibinfo{title}{{The Heisenberg Representation
  of Quantum Computers}}}.
\newblock In: {\sl \bibinfo{booktitle}{{Group22: Proceedings of the XXII
  International Colloquium on Group Theoretical Methods in Physics}}}, pp.
  \bibinfo{pages}{32--43}.
\newblock \urlprefix\url{https://arxiv.org/abs/quant-ph/9807006}.

\bibitemdeclare{inproceedings}{Green2013}
\bibitem{Green2013}
\bibinfo{author}{Alexander~S. \surnamestart Green\surnameend},
  \bibinfo{author}{Peter~LeFanu \surnamestart Lumsdaine\surnameend},
  \bibinfo{author}{Neil~J. \surnamestart Ross\surnameend},
  \bibinfo{author}{Peter \surnamestart Selinger\surnameend} \&
  \bibinfo{author}{Beno{\^{i}}t \surnamestart Valiron\surnameend}
  (\bibinfo{year}{2013}): \emph{\bibinfo{title}{Quipper: A Scalable Quantum
  Programming Language}}.
\newblock In: {\sl \bibinfo{booktitle}{Proc. PLDI '13}}, pp.
  \bibinfo{pages}{333--342}, \doi{10.1145/2491956.2462177}.
\newblock \urlprefix\url{https://arxiv.org/abs/1304.3390}.

\bibitemdeclare{inproceedings}{Honda2015}
\bibitem{Honda2015}
\bibinfo{author}{Kentaro \surnamestart Honda\surnameend}
  (\bibinfo{year}{2015}): \emph{\bibinfo{title}{Analysis of Quantum
  Entanglement in Quantum Programs using Stabilizer Formalism}}.
\newblock In: {\sl \bibinfo{booktitle}{Proc. QPL '15}}, pp.
  \bibinfo{pages}{262--272}, \doi{10.4204/EPTCS.195.19}.

\bibitemdeclare{article}{knill2005quantum}
\bibitem{knill2005quantum}
\bibinfo{author}{E.~\surnamestart Knill\surnameend} (\bibinfo{year}{2005}):
  \emph{\bibinfo{title}{Quantum computing with realistically noisy devices}}.
\newblock {\sl \bibinfo{journal}{Nature}}
  \bibinfo{volume}{434}(\bibinfo{number}{7029}), pp. \bibinfo{pages}{39--44},
  \doi{10.1038/nature03350}.
\newblock \urlprefix\url{https://arxiv.org/abs/quant-ph/0410199}.

\bibitemdeclare{article}{van2005local}
\bibitem{van2005local}
\bibinfo{author}{Maarten \surnamestart Van~den Nest\surnameend},
  \bibinfo{author}{Jeroen \surnamestart Dehaene\surnameend} \&
  \bibinfo{author}{Bart \surnamestart De~Moor\surnameend}
  (\bibinfo{year}{2005}): \emph{\bibinfo{title}{Local unitary versus local
  Clifford equivalence of stabilizer states}}.
\newblock {\sl \bibinfo{journal}{Phys. Rev. A}} \bibinfo{volume}{71}, p.
  \bibinfo{pages}{062323}, \doi{10.1103/PhysRevA.71.062323}.
\newblock \urlprefix\url{https://arxiv.org/abs/quant-ph/0411115}.

\bibitemdeclare{book}{Nielsen2010}
\bibitem{Nielsen2010}
\bibinfo{author}{Michael~A. \surnamestart Nielsen\surnameend} \&
  \bibinfo{author}{Isaac~L. \surnamestart Chuang\surnameend}
  (\bibinfo{year}{2010}): \emph{\bibinfo{title}{Quantum Computation and Quantum
  Information: 10th Anniversary Edition}}.
\newblock \bibinfo{publisher}{Cambridge University Press},
  \doi{10.1017/CBO9780511976667}.

\bibitemdeclare{article}{paetznick2013universal}
\bibitem{paetznick2013universal}
\bibinfo{author}{Adam \surnamestart Paetznick\surnameend} \&
  \bibinfo{author}{Ben~W. \surnamestart Reichardt\surnameend}
  (\bibinfo{year}{2013}): \emph{\bibinfo{title}{Universal Fault-Tolerant
  Quantum Computation with Only Transversal Gates and Error Correction}}.
\newblock {\sl \bibinfo{journal}{Phys. Rev. Lett.}} \bibinfo{volume}{111}, p.
  \bibinfo{pages}{090505}, \doi{10.1103/PhysRevLett.111.090505}.
\newblock \urlprefix\url{https://arxiv.org/abs/1304.3709}.

\bibitemdeclare{article}{Perdrix2007}
\bibitem{Perdrix2007}
\bibinfo{author}{Simon \surnamestart Perdrix\surnameend}
  (\bibinfo{year}{2007}): \emph{\bibinfo{title}{Quantum Patterns and Types for
  Entanglement and Separability}}.
\newblock {\sl \bibinfo{journal}{Electron. Notes Theor. Comput. Sci.}}
  \bibinfo{volume}{170}, pp. \bibinfo{pages}{125--138},
  \doi{10.1016/j.entcs.2006.12.015}.
\newblock \bibinfo{note}{Proc. QPL '05}.

\bibitemdeclare{inproceedings}{Perdrix2008}
\bibitem{Perdrix2008}
\bibinfo{author}{Simon \surnamestart Perdrix\surnameend}
  (\bibinfo{year}{2008}): \emph{\bibinfo{title}{Quantum Entanglement Analysis
  Based on Abstract Interpretation}}.
\newblock In: {\sl \bibinfo{booktitle}{Static Analysis}}, pp.
  \bibinfo{pages}{270--282}, \doi{10.1007/978-3-540-69166-2_18}.
\newblock \urlprefix\url{https://arxiv.org/abs/0801.4230}.

\bibitemdeclare{book}{Pierce:TypeSystems}
\bibitem{Pierce:TypeSystems}
\bibinfo{author}{Benjamin~C. \surnamestart Pierce\surnameend}
  (\bibinfo{year}{2002}): \emph{\bibinfo{title}{Types and Programming
  Languages}}.
\newblock \bibinfo{publisher}{MIT Press}.

\bibitemdeclare{inproceedings}{Rand2018}
\bibitem{Rand2018}
\bibinfo{author}{Robert \surnamestart Rand\surnameend},
  \bibinfo{author}{Jennifer \surnamestart Paykin\surnameend},
  \bibinfo{author}{Dong-Ho \surnamestart Lee\surnameend} \&
  \bibinfo{author}{Steve \surnamestart Zdancewic\surnameend}
  (\bibinfo{year}{2018}): \emph{\bibinfo{title}{Re{QWIRE}: Reasoning about
  Reversible Quantum Circuits}}.
\newblock In: {\sl \bibinfo{booktitle}{Proc. QPL '18}}, pp.
  \bibinfo{pages}{299--312}, \doi{10.4204/EPTCS.287.17}.

\bibitemdeclare{article}{Sabelfeld2003}
\bibitem{Sabelfeld2003}
\bibinfo{author}{A.~\surnamestart Sabelfeld\surnameend} \&
  \bibinfo{author}{A.~C. \surnamestart Myers\surnameend}
  (\bibinfo{year}{2006}): \emph{\bibinfo{title}{Language-based information-flow
  Security}}.
\newblock {\sl \bibinfo{journal}{IEEE J. Sel. Areas Commun.}}
  \bibinfo{volume}{21}(\bibinfo{number}{1}), pp. \bibinfo{pages}{5--19},
  \doi{10.1109/JSAC.2002.806121}.
\newblock
  \urlprefix\url{https://www.cs.cornell.edu/andru/papers/jsac/sm-jsac03.pdf}.

\bibitemdeclare{article}{steane1997active}
\bibitem{steane1997active}
\bibinfo{author}{A.~M. \surnamestart Steane\surnameend} (\bibinfo{year}{1997}):
  \emph{\bibinfo{title}{Active Stabilization, Quantum Computation, and Quantum
  State Synthesis}}.
\newblock {\sl \bibinfo{journal}{Phys. Rev. Lett.}} \bibinfo{volume}{78}, pp.
  \bibinfo{pages}{2252--2255}, \doi{10.1103/PhysRevLett.78.2252}.
\newblock \urlprefix\url{https://arxiv.org/abs/quant-ph/9611027}.

\bibitemdeclare{inproceedings}{Svore2018}
\bibitem{Svore2018}
\bibinfo{author}{Krysta \surnamestart Svore\surnameend}, \bibinfo{author}{Alan
  \surnamestart Geller\surnameend}, \bibinfo{author}{Matthias \surnamestart
  Troyer\surnameend}, \bibinfo{author}{John \surnamestart Azariah\surnameend},
  \bibinfo{author}{Christopher \surnamestart Granade\surnameend},
  \bibinfo{author}{Bettina \surnamestart Heim\surnameend},
  \bibinfo{author}{Vadym \surnamestart Kliuchnikov\surnameend},
  \bibinfo{author}{Mariia \surnamestart Mykhailova\surnameend},
  \bibinfo{author}{Andres \surnamestart Paz\surnameend} \&
  \bibinfo{author}{Martin \surnamestart Roetteler\surnameend}
  (\bibinfo{year}{2018}): \emph{\bibinfo{title}{Q\#: Enabling Scalable Quantum
  Computing and Development with a High-level DSL}}.
\newblock In: {\sl \bibinfo{booktitle}{Proc. Real World Domain Specific
  Languages Workshop (RWDSL) 2018}}, pp. \bibinfo{pages}{7:1--7:10},
  \doi{10.1145/3183895.3183901}.
\newblock \urlprefix\url{https://arxiv.org/abs/1803.00652}.

\bibitemdeclare{misc}{yoder2017universal}
\bibitem{yoder2017universal}
\bibinfo{author}{Theodore~J. \surnamestart Yoder\surnameend}
  (\bibinfo{year}{2017}): \emph{\bibinfo{title}{Universal fault-tolerant
  quantum computation with Bacon-Shor codes}}.
\newblock \bibinfo{howpublished}{arXiv preprint}.
\newblock \urlprefix\url{https://arxiv.org/abs/1705.01686}.

\end{thebibliography}

\end{document}